\documentclass[final]{siamltex}
\usepackage{amsmath}
\usepackage{mathrsfs}
\usepackage{amsfonts}
\usepackage[dvips]{graphicx}
\usepackage{psfrag}
\usepackage{amsopn}
\usepackage{amstext}
\usepackage{amssymb}
\usepackage{amscd}
\usepackage{latexsym}

\DeclareMathOperator{\T}{T}

\DeclareMathOperator{\cone}{cone}
\DeclareMathOperator{\diff}{d}

\DeclareMathOperator{\conv}{conv}

\DeclareMathOperator{\probability}{P}
\DeclareMathOperator{\expect}{E}
\DeclareMathOperator{\1}{1}

\newcommand{\fF}{\frak F}
\newcommand{\fG}{\frak G}
\newcommand{\fS}{\frak S}
\newcommand{\cM}{\mathscr M}
\newcommand{\cB}{\mathscr B}
\newcommand{\calF}{\mathcal F}
\newcommand{\calG}{\mathcal G}
\newcommand{\calS}{\mathcal S}
\newcommand{\calC}{\mathcal C}
\newcommand{\comment}[1]{}

\newcommand{\N}{\mathbb{N}}
\newcommand{\R}{\mathbb{R}}

\begin{document}

\title{A General Duality Relation with Applications in Quantitative Risk Management}

\author{Raphael Hauser\footnotemark[1], Sergey Shahverdyan \footnotemark[2] and Paul Embrechts\footnotemark[3]}

\date{\today}

\renewcommand{\thefootnote}{\fnsymbol{footnote}}
\footnotetext[2]{Mathematical Institute, Oxford University,
Andrew Wiles Building, Radcliffe Observatory Quarter, Woodstock Road, Oxford OX2 6GG,
United Kingdom, {\ttfamily hauser@maths.ox.ac.uk}. Associate Professor in Numerical Mathematics, and Tanaka Fellow in Applied Mathematics at Pembroke College, Oxford. This author was supported through grant EP/H02686X/1 from the Engineering and Physical Sciences Research Council of the UK. }
\footnotetext[2]{Mathematical Institute, Oxford University,
Andrew Wiles Building, Radcliffe Observatory Quarter, Woodstock Road, Oxford OX2 6GG,
United Kingdom, {\ttfamily sergey.shahverdyan@maths.ox.ac.uk}. This author was supported through grant EP/H02686X/1 from the Engineering and Physical Sciences Research Council of the UK. }
\footnotetext[3]{Department of Mathematics, ETH Zurich, 8092 Zurich, Switzerland 
{\ttfamily embrechts@math.ethz.ch}. Senior SFI Professor.}

\renewcommand{\thefootnote}{\arabic{footnote}}

\maketitle

\begin{abstract}A fundamental problem in risk management is the robust aggregation of different sources of risk in a situation where little or no data are available to infer information about their dependencies. A popular approach to solving this problem is to formulate an optimization problem under which one maximizes a risk measure over all multivariate distributions that are consistent with the available data. In several special cases of such models, there exist dual problems that are easier to solve or approximate, yielding robust bounds on the aggregated risk. In this chapter we formulate a general optimization problem, which can be seen as a doubly infinite linear programming problem, and we show that the associated dual generalizes several well known special cases and extends to new risk management models we propose.
\end{abstract}

\section{Introduction}\label{introduction}

An important problem in quantitative risk management is to aggregate 
several individually studied types of risks into an overall position. Mathematically, 
this translates into studying the worst-case distribution tails of $\Psi(X)$, 
where $\Psi:\R^n\rightarrow\R$ is a given function that represents the risk (or underiability) of an 
outcome, and where $X$ is a random vector that takes values in $\R^n$ and whose distribution is only partially 
known. For example, one may only have information about the marginals of $X$ and possibly partial 
information about some of the moments. 

To solve such problems, duality is often exploited, as the dual 
may be easier to approach numerically or analytically \cite{puccetti1, puccetti2,  puccetti3, 
ramachandran, popescu}. Being able to formulate a dual is also important in cases where the primal is approachable algorithmically, as solving the primal and dual problems jointly provides an approximation guarantee throughout the run of a solve: if the duality gap (the difference between the primal and dual objective values) falls below a chosen threshold relative to the primal objective, the algorithm can be stopped with a guarantee of approximating the optimum to a fixed precision that depends on the chosen threshold. This is a well-known technique in convex optimization, see e.g.\ \cite{borwein-lewis}. 

Although for some special cases of the marginal 
problem analytic solutions and powerful numerical heuristics exist \cite{puccettinew1, puccettinew2, puccettinew3, wang1, wang2}, these techniques do not apply when additional constraints are imposed to force the probability measures over which we maximize the risk to conform with empirical observations: In a typical case, the bulk of the empirical data may be contained in a region $D$ that can be approximated by an ellipsoid or the union of several (disjoint or overlapping) polyhedra. For a probability measure $\mu$ to be considered a reasonable explanation of the true distribution of (multi-dimensional) losses, one would require the probability mass contained in $D$ to lie in an empirically estimated confidence region, that is, $\ell\leq\mu(D)\leq u$ for some estimated bounds $\ell<u$. In such a situation, the derivation of robust risk 
aggregation bounds via dual problems remains a powerful and interesting approach. 

In this chapter we formulate a general optimization problem, which can be seen as a doubly infinite linear programming problem, and we show that the associated dual generalizes several well known special cases. 
We then apply this duality framework to a new class of risk management models we propose in Section \ref{bounds on integrals}. 

\section{A General Duality Relation}\label{duality}

Let $(\Phi,\fF)$, $(\Gamma,\fG)$ and $(\Sigma,\fS)$ 
be complete measure spaces, and let 
$A:\,\Gamma\times\Phi\rightarrow\R$, 
$a:\,\Gamma\rightarrow\R$, 
$B:\,\Sigma\times\Phi\rightarrow\R$, 
$b:\,\Sigma\rightarrow\R$, and 
$c:\,\Phi\rightarrow\R$ 
be bounded measurable functions on these spaces and the corresponding product 
spaces. Let $\cM_{\fF}$, $\cM_{\fG}$ and $\cM_{\fS}$ be the set of signed 
measures with finite variation on $(\Phi,\fF)$, $(\Gamma,\fG)$ and 
$(\Sigma,\fS)$ respectively. 
We now consider the following pair of optimization problems over $\cM_{\fF}$ and 
$\cM_{\fG}\times\cM_{\fS}$ respectively, 
\begin{align*}
\text{(P)}\quad\sup_{\calF\in\cM_{\fF}}\,&\int_{\Phi}c(x)\diff \calF(x)\\
\text{s.t. }&\int_{\Phi}A(y,x)\diff \calF(x)\leq a(y),\quad(y\in\Gamma),\\
&\int_{\Phi}B(z,x)\diff \calF(x)= b(z),\quad(z\in\Sigma),\\
&\calF\geq 0,
\end{align*}
and 
\begin{align*}
\text{(D)}\quad\inf_{(\calG,\calS)\in\cM_{\fG}\times\cM_{\fS}}
\,&\int_{\Gamma}a(y)\diff \calG(y)+\int_{\Sigma}b(z)\diff\calS(z),\\
\text{s.t. }&\int_{\Gamma}A(y,x)\diff\calG(y)+\int_{\Sigma}
B(z,x)\diff\calS(z)\geq c(x),\quad(x\in\Phi),\\
&\calG\geq 0.
\end{align*}
We claim that the infinite-programming problems (P) and (D) are duals of 
each other. 

\begin{theorem}[Weak Duality]\label{weak duality}
For every (P)-feasible measure $\calF$ and every 
(D)-feasible pair $(\calG,\calS)$ we have 
\begin{equation*}
\int_{\Phi}c(x)\diff\calF(x)\leq \int_{\Gamma}a(y)\diff\calG(y)+
\int_{\Sigma}b(z)\diff\calS(z).
\end{equation*}
\end{theorem}

\begin{proof}
Using Fubini's Theorem, we have 
\begin{align*}
\int_{\Phi}c(x)\diff\calF(x)&\leq
\int_{\Gamma\times\Phi}A(y,x)\diff(\calG\times\calF)(y,x)+
\int_{\Sigma\times\Phi}B(z,x)\diff(\calS\times\calF)(z,x)\\
&\leq\int_{\Gamma}a(y)\diff\calG(y)+
\int_{\Sigma}b(z)\diff\calS(z).
\end{align*}
\end{proof}

In various special cases, such as those discussed in Section \ref{special cases}, strong duality is known to hold subject to regularity assumptions, that is, the optimal values of (P) and (D) coincide. Another special case under  which strong duality applies is when the measures $\calF$, $\calG$ and $\calS$ have densities in appropriate Hilbert spaces, see the forthcoming DPhil thesis of the second author \cite{sergey}. 

We remark that the quantifiers in the constraints can be weakened if the set of allowable measures is restricted. For example, if $\calG$ is restricted to lie in a set of measures that are absolutely continuous with respect to a fixed measure $\calG_0\in\cM_{\fG}$, then the quantifier $(y\in\Gamma)$ can be weakened to $(\calG_0\text{-almost all }y\in\Gamma)$. 

\section{Classical Examples}\label{special cases}

Our general duality relation of Theorem \ref{weak duality} generalizes many classical duality results, of which we now point out a few examples. Let $p(x_1,\dots,x_k)$ be a function of $k$ arguments. Then we write 
\begin{equation*}
1_{\{x: p(x)\geq 0\}}:=1_{\{y: p(y)\geq 0\}}(x)=\begin{cases}1\quad&\text{if }p(x)\geq 0,\\
0\quad&\text{otherwise}. 
\end{cases}
\end{equation*}
In other words, we write the argument $x$ of the indicator function directly into the set $\{y: p(y)\geq 0\}$ 
that defines the function, rather than using a separate set of variables $y$. This abuse of notation will make it easier to identify which inequality is satisfied by the arguments where the function $1_{\{y: p(y)\geq 0\}}(x)$ takes the value $1$. 

We start with the Moment Problem studied by Bertsimas \& Popescu \cite{popescu}, who considered generalized Chebychev inequalities of the form 
\begin{align*}
\text{(P')}\quad\sup_{X}\;&\probability[r(X)\geq 0]\\
\text{s.t. }&\expect_{\mu}[X_1^{k_1}\dots X_n^{k_n}]=b_{k},\quad(k\in J),\\
&X\text{ a random vector taking values in }\R^n,\nonumber
\end{align*}
where $r:\R^n\rightarrow\R$ is a multivariate polynomial and $J\subset
\N^n$ is a finite sets of multi-indices. In other words, some moments of $X$ are known. 
By choosing $\Phi=\R^n$, 
$\Gamma=\emptyset$, $\Sigma=J\cup\{0\}$, 
\begin{align*}
&B(k,x)=x_1^{k_1}\dots x_n^{k_n},\quad b(k)=b_k,\quad(k\in J),\\
&B(0,x)=\1_{\R^n},\quad b(0)=1,
\end{align*}
and $c(x)=\1_{\{x:\,r(x))\geq 0\}}$, where we made use of the abuse of notation discussed above. 
Problem (P') becomes 
a special case of the primal problem considered in Section \ref{duality}, 
\begin{align*}
\text{(P)}\quad\sup_{\calF}\,&\int_{\R^n}\1_{\{x:\,r(x)\geq 0\}}\diff\calF(x)\\
\text{s.t. }&\int_{\R^n}x_1^{k_1}\dots x_n^{k_n}\diff\calF(x)=b_{k},\quad(k\in J),\\
&\int_{\R^n}1\diff\calF(x)=1,\\
&\calF\geq 0.
\end{align*}
Our dual 
\begin{align*}
\text{(D)}\quad\inf_{(z,z_0)\in\R^{|J|+1}}\,&\sum_{k\in J} z_k b_k + z_0\\
\text{s.t. }&\sum_{k\in J} z_k x_1^{k_1}\dots x_n^{k_n} + z_0\geq \1_{\{x: r(x)\geq 0\}},
\quad(x\in\R^n)
\end{align*}
is easily seen to be identical with the dual (D') identified by Bertsimas 
\& Popescu, 
\begin{align*}
\text{(D')}\quad\inf_{(z,z_0)\in\R^{|J|+1}}\,&\sum_{k\in J} z_k b_k + z_0\\
\text{s.t. }&\forall\,x\in\R^n, r(x)\geq 0\Rightarrow
\sum_{k\in J} z_k x_1^{k_1}\dots x_n^{k_n} + z_0-1\geq 0,\\
&\forall\,x\in\R^n, \sum_{k\in J} z_k x_1^{k_1}\dots x_n^{k_n} + z_0\geq 0.
\end{align*}
Note that since $\Gamma,\Sigma$ are finite, the constraints of (D') are 
polynomial copositivity constraints. The numerical solution of 
semi-infinite programming problems 
of this type can be approached via a nested hierarchy of semidefinite 
programming relaxations that yield better and better approximations to 
(D'). The highest level problem within this hierarchy is guaranteed to 
solve (D') exactly, although the corresponding SDP is of exponential 
size in the dimension $n$, in the degree of the polymomial $r$, and 
in $\max_{k\in J}(\sum_i k_i)$. For further details see \cite{popescu, parrilo, lasserre}, 
and Section \ref{piecewise polynomial} below.

Next, we consider the Marginal Problem studied by R\"uschendorf \cite{ruschendorf1, ruschendorf2} and 
Ramachandran \& R\"uschendorf \cite{ramachandran}, 
\begin{equation*}
\text{(P')}\quad\sup_{\calF\in\cM_{F_1,\dots,F_n}}\,\int_{\R^n}h(x)\diff
\calF(x), 
\end{equation*}
where $\cM_{F_1,\dots,F_n}$ is the set of probability measures on 
$\R^n$ whose marginals 
have the cdfs $F_i$ $(i=1,\dots,n)$.  Problem (P') can easily be seen 
as a special case of the framework of Section \ref{duality} by 
setting $c(x)=h(x)$, $\Phi=\R^n$, $\Gamma=\emptyset$, $\Sigma=\N_n\times\R$, 
$B(i,z,x)=\1_{\{y:\,y_i\leq z\}}$ (using the abuse of notation discussed earlier), 
and $b_i(z)=F_i(z)$ $(i\in\N_n,\,z\in\R)$, 
\begin{align*}
\text{(P)}\quad\sup_{\calF}\,&\int_{\R^n}h(x)\diff\calF(x)\\
\text{s.t. }&\int_{\R}\1_{\{x_i\leq z\}}\diff\calF(x)
=F_i(z),\quad(z\in\R, i\in\N_n)\\
&\calF\geq 0. 
\end{align*}
Taking the dual, we find 
\begin{align*}
\text{(D)}\quad\inf_{\calS_1,\dots,\calS_n}\,&
\sum_{i=1}^n\int_{\R}F_i(z)\diff\calS_i(z)\\
\text{s.t. }&\sum_{i=1}^n\int_{\R}\1_{\{x_i\leq z\}}\diff\calS_i(z)
\geq h(x),\quad(x\in\R^n).
\end{align*}
The signed measures $\calS_i$ being of finite variation, the functions 
$S_i(z)=\calS((-\infty,z])$ and the limits $s_i=\lim_{z\rightarrow
\infty}S_i(z)=\calS((-\infty,+\infty))$ are well defined and finite. 
Furthermore, using $\lim_{z\rightarrow-\infty}F_i(z)=0$ and 
$\lim_{z\rightarrow+\infty}F_i(z)=1$, we have 
\begin{align*}
\sum_{i=1}^n\int_{\R}F_i(z)\diff\calS(z)&=
\sum_{i=1}^n\left(F_i(z)S_i(z)|^{+\infty}_{-\infty}-\int_{\R}S_i(z)\diff F_i(z)
\right)\\
&=\sum_{i=1}^n s_i - \sum_{i=1}^n\int_{\R}S_i(z)\diff F_i(z)\\
&=\sum_{i=1}^n\int_{\R}(s_i-S_i(z))\diff F_i(z),
\end{align*}
and likewise, 
\begin{equation*}
\sum_{i=1}^n\int_{\R}\1_{\{x_i\leq z\}}\diff\calS_i(z)=
\sum_{i=1}^n\int_{x_i}^{+\infty}1\diff\calS_i(z)
=\sum_{i=1}^n (s_i - S_i(x_i)). 
\end{equation*}
Writing $h_i(z)=s_i-S_i(z)$, (D) is therefore equivalent to 
\begin{align*}
\text{(D')}\quad\inf_{h_1,\dots,h_n}\,&\sum_{i=1}^n\int_{\R}h_i(z)\diff F_i(z)\\
\text{s.t. }&\sum_{i=0}^n h_i(x_i)\geq h(x),\quad
(x\in\R^n).
\end{align*}
This is the dual identified by Ramachandran \& R\"uschendorf \cite{ramachandran}. 
Due to the general form of the functions $h_i$, the infinite programming problem 
(D') is not directly usable in numerical computations. However, for specific 
$h(x)$, (D')-feasible functions $(h_1,\dots,h_n)$ can sometimes be constructed 
explicitly, yielding an upper bound on the optimal objective function value of 
(P') by virtue of Theorem \ref{weak duality}. Embrechts \& Puccetti 
\cite{puccetti1,puccetti2,puccetti3} used this approach to derive quantile bounds 
on $X_1+\dots+X_n$, where $X$ is a random vector with known marginals 
but unknown joint distribution. In this case, the relevant primal objective function 
is defined by $h(x)=\1_{\{x:\,e^{\T}x\geq t\}}$, where $t\in\R$ is a fixed level. 
More generally, $h(x)=\1_{\{x:\,\Psi(x)\geq t\}}$ can be chosen, where $\Psi$ 
is a relevant risk aggregation function, or $h(x)$ can model any risk measure 
of choice. 

Our next example is the Marginal Problem with Copula Bounds, an extension to the marginal problem mentioned in \cite{puccetti1}. The copula defined by the probability measure $\calF$ with marginals 
$F_i$ is the function 
\begin{align*}
\calC_{\calF}:\,[0,1]^n&\rightarrow[0,1],\\
u&\mapsto F\left(F_1^{-1}(u_1),\dots,F_n^{-1}(u_n)\right). 
\end{align*}
A copula is any function $\calC:[0,1]^n\rightarrow[0,1]$ that satisfies 
$\calC=\calC_{\calF}$ for some probability measure $\calF$ on $\R^n$. 
Equivalently, a copula is the multivariate cdf of any probability measure 
on the unit cube $[0,1]^n$ with uniform marginals. In 
quantitative risk management, using the model 
\begin{equation*}
\quad\sup_{\calF\in\cM_{F_1,\dots,F_n}}\,\int_{\R^n}h(x)\diff\calF(x) 
\end{equation*}
to bound the worst case risk for a random vector $X$ 
with marginal distributions $F_i$ can be overly 
conservative, as no dependence structure between the coordinates 
of $X_i$ is assumed given at all. The structure that determines this dependence 
being the copula $\calC_{\calF}$, where $\calF$ is the multivariate 
distribution of $X$, Embrechts \& Puccetti \cite{puccetti1} suggest problems of the form 
\begin{align*}
\text{(P')}\quad\sup_{\mu\in\cM_{F_1,\dots,F_n}}\,&\int_{\R^n}h(x)\diff\mu(x), \\
\text{s.t. }&\calC_{\text{lo}}\leq\calC_{\calF}\leq\calC_{\text{up}},
\end{align*}
as a natural framework to study the situation in which partial dependence 
information is available. In problem (P'), $\calC_{\text{lo}}$ and $\calC_{\text{up}}$ are 
given copulas, and inequality between copulas is defined by pointwise 
inequality,
\begin{equation*}
\calC_{\text{lo}}(u)\leq\calC_{\calF}(u)\quad(u\in[0,1]^n).
\end{equation*}
Once again, (P') is a special case of the general framework studied in 
Section \ref{duality}, as it is equivalent to write 
\begin{align*}
\text{(P)}\quad\sup_{\calF}\,&\int_{\R^n}h(x)\diff\calF(x)\\
\text{s.t. }&\int_{\R^n}\1_{\{x\leq(F_1^{-1}(u_1),\dots,F_n^{-1}(u_n))\}}
(u,x)\diff\calF(x)\leq\calC_{\text{up}}(u),\quad(u\in[0,1]^n),\\
&\int_{\R^n}-\1_{\{x\leq(F_1^{-1}(u_1),\dots,F_n^{-1}(u_n))\}}
(u,x)\diff\calF(x)\leq-\calC_{\text{lo}}(u),\quad(u\in[0,1]^n),\\
&\int_{\R^n}\1_{\{x_i\leq z\}} (z,x)\diff\calF(x)=F_i(z),\quad(i\in\N_n,\,z\in\R),\\
&\calF\geq 0. 
\end{align*}
The dual of this problem is given by 
\begin{align*}
\text{(D)}\quad\inf_{\calG_{\text{up}},\calG_{\text{lo}},\calS_1,\dots,\calS_n}\,&
\int_{[0,1]^n}\calC_{\text{up}}(u)\diff\calG_{\text{up}}(u)-
\int_{[0,1]^n}\calC_{\text{lo}}(u)\diff\calG_{\text{lo}}(u)
+\sum_{i=1}^n\int_{\R}F_i(z)\diff\calS_i(z)\\
\text{s.t. }&\int_{[0,1]^n}\1_{\{x\leq(F_1^{-1}(u_1),\dots,F_n^{-1}(u_n))\}}
(u,x)\diff\calG_{\text{up}}(u)\\
&\hspace{1cm}-\int_{[0,1]^n}\1_{\{x\leq(F_1^{-1}(u_1),\dots,F_n^{-1}(u_n))\}}
(u,x)\diff\calG_{\text{lo}}(u)\\
&\hspace{1cm}+\sum_{i=1}^n\int_{\R}\1_{\{x_i\leq z\}}\diff\calS_i(z)
\geq h(x),\quad(x\in\R^n),\\
&\calG_{\text{lo}},\calG_{\text{up}}\geq 0.
\end{align*}
Using the notation $s_i, S_i$ introduced in Section \ref{special cases}, this 
problem can be written as 
\begin{align*}
\inf_{\calG_{\text{up}},\calG_{\text{lo}},\calS_1,\dots\calS_n}\,&
\int_{[0,1]^n}\calC_{\text{up}}(u)\diff\calG_{\text{up}}(u)-
\int_{[0,1]^n}\calC_{\text{lo}}(u)\diff\calG_{\text{lo}}(u)
+\sum_{i=1}^n\int_{\R}(s_i-S_i(z))\diff F_i(z)\\
\text{s.t. }&\calG_{\text{up}}(\cB(x))
-\calG_{\text{lo}}(\cB(x))
+\sum_{i=1}^n(s_i-S_i(x_i))\geq h(x),\quad(x\in\R^n),\\
&\calG_{\text{up}},\calG_{\text{lo}}\geq 0, 
\end{align*}
where $\cB(x)=\{u\in[0,1]^n:\,u\geq(F_1(x_1),\dots,F_n(x_n))\}$. To the best of our knowledge, this dual has not been identified before.

Due to the high dimensionality of the space of variables and constraints 
both in the primal and dual, the marginal problem with copula bounds  
is difficult to solve numerically, even for very coarse disrecte approximations.

\section{Robust Risk Aggregation via Bounds on Integrals}
\label{bounds on integrals}

In quantitative risk management, distributions are often estimated within a parametric family from the available data. For example, the tails of marginal distributions may be estimated via extreme value theory, or a Gaussian copula may be fitted to the multivariate distribution of all risks under consideration, to model their dependencies. The choice of a parametric family introduces {\em model uncertainty}, while fitting a distribution from this family via statistical estimation introduces {\em parameter uncertainty}. In both cases, a more robust alternative would be to study models in which the available data is only used to estimate upper and lower bounds on finitely many integrals of the form 
\begin{equation}\label{of the form}
\int_{\Phi}\phi(x)\diff\calF(x),
\end{equation}
where $\phi(x)$ is a suitable test function. A suitable way of estimating upper and lower bounds on such integrals from sample data $x_i$ $(i\in\N_k)$ is to estimate confidence bounds via bootstrapping.  

\subsection{Motivation}\label{motivation}

To motivate the use of constraints in the form of bounds on integrals \eqref{of the form}, we offer the following explanations: First of all, discretized marginal constraints are of this form with piecewise constant test functions, as the requirement that $F_i(\xi_{k})-F_i(\xi_{k-1})=b_k$ $(k=1,\dots,\ell)$ for a fixed set of discretization points $\xi_0<\dots<\xi_{\ell}$ can be expressed as 
\begin{equation}\label{new form}
\int_{\Phi}1_{\{\xi_{k}\leq x_i\leq\xi_{k-1}\}}\diff\calF(x)=b_k,\quad(k=1,\dots,\ell). 
\end{equation}
It is furthermore quite natural to relax each of these equality constraints to two inequality constraints 
\begin{equation*}
b^{\ell}_{k,i}\leq\int_{\Phi}1_{\{\xi_{k}\leq x_i\leq\xi_{k-1}\}}\diff\calF(x)\leq b^u_{k,i} 
\end{equation*}
when $b_k$ is estimated from data. 

More generally, constraints of the form $\probability[X\in S_j]\leq b^u_j$ for some measurable 
$S_j\subseteq\R^n$ of interest can be written as 
\begin{equation*}
\int_{\Phi} 1_{S_j}(x)\diff\calF(x)\leq b^u_j. 
\end{equation*}
A collection of $\ell$ constraints of this form can be relaxed by replacing them by a convex combination 
\begin{equation*}
\int_{\Phi}\sum_{j=1}^{\ell}w_j 1_{S_j}(x)\diff\calF(x)\leq\sum_{j=1}^{\ell} w_j b^u_j,  
\end{equation*}
where the weights $w_j>0$ satisfy $\sum_j w_j =1$ and express the relative importance of each constituent constraint. Non-negative test functions thus have a natural interpretation as importance densities in sums-of-constraints relaxations. This allows one to put higher focus on getting the probability mass right in regions where it particularly matters (e.g., values of $X$ that account for the bulk of the profits of a financial institution), while maximzing the risk in the tails without having to resort to too fine a discretization.   

While this suggests to use a piecewise approximation of a prior estimate of the density of $X$ as a test function, the results are robust under mis-specification of this prior, for as long as $\phi(x)$ is nonconstant, constraints that involve the integral \eqref{of the form} tend to force the probability weight of $X$ into the regions where the sample points are denser. To illustrate this, consider a univariate random variable with density $f(x)=2/3(1+x)$ on $x\in[0,1]$ and test function $\phi(x)=1+a x$ with $a\in[-1,1]$. Then 
$\int_{0}^{1}\phi(x) f(x)\diff x=1+5a/9$.
The most dispersed probability measure on $[0,1]$ that satisfies 
\begin{equation}\label{hallo}
\int_{0}^{1}\phi(x)\diff\calF(x)=1+\frac{5a}{9}
\end{equation}
has an atom of weight $4/9$ at $0$ and an atom of weight $5/9$ at $1$ independently of $a$, as long as $a\neq 0$. The constraint \eqref{hallo} thus forces more probability mass into the right half of the interval $[0,1]$, where the unknown (true) density $f(x)$ has more mass and produces more sample points. 

As a second illustration, take the density $f(x)=3 x^2$ and the same linear test function as above. This time we find $\int_{0}^{1}\phi(x) f(x)\diff x=1+3a/4$,  and the most dispersed probability measure on $[0,1]$ that satisfies 
\begin{equation*}
\int_{0}^{1}\phi(x)\diff\calF(x)=1+\frac{3a}{4}
\end{equation*}
has an atom of weight $3/4$ at $0$ and an atom of weight $1/4$ at $1$ independently of $a\neq 0$, with similar conclusions as above, except that the effect is even stronger, correctly reflecting the qualitative features of the density $f(x)$.

\subsection{General Setup and Duality}\label{general setup}

Let $\Phi$ be decomposed into a partition $\Phi=\bigcup_{i=1}^k\Xi_i$ of polyhedra $\Xi_i$ with nonempty interior, chosen as regions in which a reasonable number of data points are available to estimate integrals of the form \eqref{of the form}.  

Each polyhedron has a primal description in terms of generators, 
\begin{equation*}
\Xi_i=\conv(q^i_1,\dots,q^i_{n_i})+\cone(r^i_1,\dots,r^i_{o_i})
\end{equation*}
where $\conv(q^i_1,\dots,q^i_{n_i})$ is the polytope with vertices $q^i_n\in\R^n$, and 
\begin{equation*}
\cone(r^i_1,\dots,r^i_{o_i})=\left\{\sum_{m=1}^{o_i}\xi_m r^i_m:\,
\xi_m\geq 0\;(m\in\N_{o_i})\right\} 
\end{equation*}
is the polyhedral cone with recession directions $r^i_m\in\R^n$. Each polyhedron also has a dual description 
in terms of linear inequalities, 
\begin{equation*}
\Xi_i=\bigcap_{j=1}^{k_i}\left\{x\in\R^n:\, \langle f^i_{j}, x\rangle\geq \ell^i_j\right\}, 
\end{equation*}
for some vectors $f^i_j\in\R^n$ and bounds $\ell^i_j\in\R$. The main case of interest is where $\Xi_i$ is either a finite or infinite box in $\R^n$ with faces parallel to the coordinate axes, or an intersection of such a box with a linear half space, in which case it is easy to pass between the primal and dual descriptions. Note 
however that the dual description is preferrable, as the description of a box in $\R^n$ requires only $2n$ linear inequalities, while the primal description requires $2^n$ extreme vertices. 

Let us now consider the problem 
\begin{align*}
\text{(P)}\quad\sup_{\calF\in\cM_{\fF}}\,&\int_{\Phi}h(x)\diff \calF(x)\\
\text{s.t. }&\int_{\Phi}\phi_{s}(x)\diff\calF(x)
\leq a_{s},\quad(s=1,\dots,M),\\
\text{s.t. }&\int_{\Phi}\psi_{t}(x)\diff\calF(x)
= b_{t},\quad(t=1,\dots,N),\\
&\int_{\Phi}1\diff\calF(x)=1,\\
&\calF\geq 0,
\end{align*}
where the test functions $\psi_t$ are piecewise linear on the partition $\Phi=\bigcup_{i=1}^k\Xi_i$, and where $-h(x)$ and the test functions $\phi_s$ are piecewise linear on the infinite polyhedra of the partition, and either jointly linear, concave, or convex on the finite polyhedra (i.e., polytopes) of the partition. The dual of (P) is 
\begin{align}
\text{(D)}\quad\inf_{(y,z)\in\R^{M+N+1}}
\,&\sum_{s=1}^{M}a_s y_s+\sum_{t=1}^{N}b_t z_t + z_0,\nonumber\\
\text{s.t. }&\sum_{s=1}^{M}y_s\phi_s(x)+\sum_{t=1}^{N}z_t\psi_t(x)+z_0\1_{\Phi}(x)
-h(x)\geq 0,\; (x\in\Phi),\label{pos0}\\
&y\geq 0.\nonumber
\end{align}

We remark that (P) is a semi-infinite programming problem with infinitely many variables and finitely many constraints, while (D) is a semi-infinite programming problem with finitely many variables and infinitely many constraints. However, the constraint \eqref{pos0} of (D) can be rewritten as copositivity requirements over the polyhedra $\Xi_i$, 
\begin{equation*}
\sum_{s=1}^{M}y_s\phi_s(x)+\sum_{t=1}^{N}z_t\psi_t(x)+z_0\1_{\Phi}(x)
-h(x)\geq 0,\quad(x\in\Xi_i), \quad(i=1,\dots,k). 
\end{equation*}
Next we will see how these copositivity constraints can be handled numerically, often by relaxing all but finitely many constraints. Nesterov's first order method can be adapted to solve the resulting problems, see \cite{nesterov1, nesterov2, sergey}. 

In what follows, we will use the notation
\begin{equation*}
\varphi_{y,z}(x)=\sum_{s=1}^{M}y_s\phi_s(x)+\sum_{t=1}^{N}z_t\psi_t(x)+z_0-h(x). 
\end{equation*}

\subsection{Piecewise Linear Test Functions}\label{piecewise linear}

The first case we discuss is when $\phi_s|_{\Xi_i}$ and $h|_{\Xi_i}$ are jointly linear. Since we 
furthermore assumed that the functions $\psi_t|_{\Xi_i}$ are linear, there exist vectors $v^i_s\in\R^n$, 
$w^i_t\in\R^n$, $g^i\in\R^n$ and constants $c^i_s\in\R$, $d^i_t\in\R$ and $e^i\in\R$ such that 
\begin{align*}
\phi_s|_{\Xi_i}(x)&=\langle v^i_s, x\rangle+c^i_s,\\
\psi_t|_{\Xi_i}(x)&=\langle w^i_t, x\rangle+d^i_t,\\
h|_{\Xi_i}(x)&=\langle g^i, x\rangle+e^i. 
\end{align*}
The copositivity condition 
\begin{equation*}
\sum_{s=1}^{M}y_s\phi_s(x)+\sum_{t=1}^{N}z_t\psi_t(x)+z_0\1_{\Phi}(x)
-h(x)\geq 0,\quad(x\in\Xi_i) 
\end{equation*}
can then be written as 
\begin{multline*}
\langle f^i_j, x\rangle \geq \ell^i_j,\quad (j=1,\dots,k_i)\Longrightarrow\\
\left\langle\sum_{s=1}^{M}y_s v^i_s+\sum_{t=1}^{N}z_t w^i_t-g^i\;,\; x\right\rangle\geq e^i-\sum_{s=1}^{M}y_s c^i_s - \sum_{t=1}^{N}z_t d^i_t - z_0. 
\end{multline*}
By Farkas' Lemma, this is equivalent to the constraints 
\begin{align}
\sum_{s=1}^{M}y_s v^i_s+\sum_{t=1}^{N}z_t w^i_t-g^i&=\sum_{j=1}^{k_i}\lambda^i_j f^i_j,\label{c1}\\
e^i-\sum_{s=1}^{M}y_s c^i_s - \sum_{t=1}^{N}z_t d^i_t - z_0&\leq\sum_{j=1}^{k_i}\lambda^i_j \ell^i_j,\label{c2}\\
\lambda^i_j&\geq 0,\quad(j=1,\dots,k_i),\label{c3}
\end{align}
where $\lambda^i_j$ are additional auxiliary decision variables. 

Thus, if all test functions are linear on all polyhedral pieces $\Xi_i$, then the dual (D) can be solved as a linear programming problem with $M+N+1+\sum_{i=1}^k k_i$ variables and $k(n+1)$ linear constraints, 
plus bound constraints on $y$ and the $\lambda^i_j$. More generally, if some but not all polyhedra correspond to jointly linear test function pieces, then jointly linear pieces can be treated as discussed above, while other pieces can be treated as discussed below. 
 
Let us briefly comment on numerical implementations, further details of which are described in the second author's thesis \cite{sergey}: An important case of the above described framework corresponds to a discretized marginal problem in which $\phi_s(x)$ are piecewise constant functions chosen as follows for $s=(i,j)$, $(\iota=1,\dots,n; j=1,\dots,m)$: Introduce $m+1$ breakpoints 
$\xi^{\iota}_0<\xi^{\iota}_1<\dots<\xi^{\iota}_m$ along each coordinate axis $\iota$, and 
consider the infinite slabs 
\begin{equation*}
S_{\iota,j}=\left\{x\in\R^n:\, \xi^{\iota}_{j-1}\leq x_{\iota}\leq\xi^{\iota}_{j}\right\}, \quad(j=1,\dots,m).
\end{equation*}
Then choose $\phi_{\iota,j}(x)=1_{S_{\iota,j}}(x)$, the indicator function of slab $S_{\iota,j}$. We remark that this approach corresponds to discretizing the constraints of the Marginal Problem described in Section \ref{special cases}, but not to discretizing the probability measures over which we maximize the aggregated risk. 

While the number of test functions is $nm$ and thus linear in the problem dimension, the number of polyhedra to consider is exponentially large, as all intersections of the form 
\begin{equation*}
\Xi_{\iota,\vec{j}}=\bigcap_{\iota=1}^n S_{\iota, j_{\iota}}
\end{equation*}
for the $m^n$ possible choices of $\vec{j}\in\N_{m}^n$ have to be treated separately. In addition, in VaR applications $h(x)$ is taken as the indicator function of an affine half space $\{x:\sum x_\iota \geq \tau\}$ for a suitably chosen threshold $\tau$, and for CVaR applications $h(x)$ is chosen as the piecewise linear function $h(x)=\max(0, \sum x_\iota -\tau)$. Thus, polyhedra $\Xi_{\iota,\vec{j}}$ that meet the affine hyperplane $\{x:\,\sum x_\iota=\tau\}$ are further sliced into two separate polyhedra. A straightforward application of the above described LP framework would thus lead to an LP with exponentially many constraints and variables. Note however,  that the constraints \eqref{c1}--\eqref{c3} now read  
\begin{align}
g^i&=\sum_{j=1}^{k_i}\lambda^i_j f^i_j,\label{c4}\\
e^i-\sum_{s=1}^{M}y_s c^i_s - z_0&\leq\sum_{j=1}^{k_i}\lambda^i_j \ell^i_j,\label{c5}\\
\lambda^i_j&\geq 0,\quad(j=1,\dots,k_i),\label{c6}
\end{align}
as $v^i_s=0$ and no test functions $\psi_t(x)$ were used, 
with $g^i=[\begin{smallmatrix}1&\dots&1\end{smallmatrix}]^{\T}$ when $\Xi_i\subseteq\{x:\,
\sum x_{\iota}\geq\tau\}$ and $g^i=0$ otherwise. That is, the vector that appears in the left-hand side of Constraint \eqref{c4} is fixed by the polyhedron $\Xi_i$ alone and does not depend on the decision variables $y,z_0$. Since $z_0$ is to be chosen as small as possible in an optimal solution of (D), the constraint \eqref{c5} has to be made as slack as possible. Therefore, the optimal values of $\lambda^i_j$ are also fixed by the polyhedron $\Xi_i$ alone and are identifiable by solving the small-scale LP 
\begin{align*}
(\lambda^{i}_j)^*=\arg\max_{\lambda}\;&\sum_{j=1}^{k_i}\lambda^{i}_j\ell^{i}_j\\
\text{s.t. }-g^i&=\sum_{j=1}^{k_i}\lambda^i_j f^i_j,\\
\lambda^i_j&\geq 0,\quad(j=1,\dots,k_i). 
\end{align*}
In other words, when the polyhedron $\Xi_i$ is considered for the first time, the variables $(\lambda^i_j)^*$ can be determined once and for all, after which the constraints \eqref{c4} -- \eqref{c6} can be replaced by 
\begin{equation*}
e^i-\sum_s y_{s}c^i_{s}-z_0\leq C_i,
\end{equation*}
where $C_i=\sum_{j=1}^{k_i}(\lambda^i_j)^* \ell^i_j$ and where the sum on the left-hand side only 
extends over the $n$ 
indices $s$ that correspond to test functions that are non-zero on $\Xi_i$. Thus, only the $nm+1$ decision variables $(y,z_0)$ are needed to solve (D). Furthermore, the exponentially many constraints correspond to an  extremely sparse constraint matrix, making the dual of (D) an ideal candidate to apply the simplex algorithm with delayed column generation. A similar approach is possible for the situation where $\phi_s$ is of the form  
\begin{equation*}
\phi_{s}(x)=1_{S_{s}}(x)\times\left(\langle v_{s}, x\rangle + c_{s}\right), 
\end{equation*}
for all $s=(\iota,j)$. The advantage of using test functions of this form is that fewer breakpoints $\xi_{\iota,j}$ are needed to constrain the distribution appropriately. 

To illustrate the power of delayed column generation, we present the following numerical experiments that compare the computation times of our standard simplex (SS) implementation and our simplex with delayed column generation (DCG) for different dimensions $d$ and numbers $k$ of polyhedra :

\begin{table}[h!]
\begin{center}
\begin{tabular}{|c|c|c|c|c|}
\hline
$d$ & $k$ & DCG & SS & DCG/SS\\
\hline
$2$ & $256$ & $2.282$ & $0.963$ & $2.340$\\
$2$ & $1'024$ & $1.077$ & $1.345$ & $0.801$\\
$2$ & $4'096$ & $3.963$ & $7.923$ & $0.500$ \\
$2$ & $16'384$ & $20.237$ & $34.495$ & $0.587$ \\
$2$ & $65'536$ & $124.89$ & $187.826$ & $0.665$ \\
\hline
$3$ & $4'096$ & $2.653$ & $13.103$ & $0.202$ \\
$3$ & $32'768$ & $42.863$ & $145.002$ & $0.300$ \\
\hline
$4$ & $4'096$ & $4.792$ & $25.342$ & $0.189$ \\
$4$ & $65'536$ & $345.6$ & $28'189.7$ & $0.012$ \\
\hline
\end{tabular}  
\end{center}
\end{table}

The results show that with increasing dimension, the savings of using delayed column generation over the standard simplex algorithm grows rapidly. Since our Matlab implementation of both methods are not optimized, the absolute computation times can be considerably improved, but the speedup achieved by using delayed column generation is representative.

\subsection{Piecewise Convex Test Functions}\label{piecewise convex}

When $\phi_s|_{\Xi_i}$ and $-h|_{\Xi_i}$ are jointly convex, then $\varphi_{y,z}(x)$ is convex. 
The copositivity constraint 
\begin{equation*}
\sum_{s=1}^{M}y_s\phi_s(x)+\sum_{t=1}^{N}z_t\psi_t(x)+z_0\1_{\Phi}(x)
-h(x)\geq 0,\quad(x\in\Xi_i) 
\end{equation*}
can then be written as 
\begin{equation*}
\langle f^i_j, x\rangle \geq \ell^i_j,\quad (j=1,\dots,k_i)\Longrightarrow
\varphi_{y,z}(x)\geq 0, 
\end{equation*}
and by Farkas' Theorem (see e.g., \cite{polik}), this condition is equivalent to 
\begin{align}
\varphi_{y,z}(x)+\sum_{j=1}^{k_i}\lambda^i_j\left(\ell^i_j-\langle f^i_j, x\rangle\right)&\geq 0,\quad(x\in\R^n),\label{global}\\
\lambda^i_j&\geq 0,\quad(j=1,\dots,k_i),\nonumber 
\end{align}
where $\lambda^i_j$ are once again auxiliary decision variables. While \eqref{global} does not reduce to finitely many constraints, the validity of this condition can be checked numerically by globally minimizing the convex function $\varphi_{y,z}(x)+\sum_{j=1}^{k_i}\lambda^i_j\left(\ell^i_j-\langle f^i_j, x\rangle\right)$.  The constraint \eqref{global} can then be enforced explicitly if a line-search method is used to solve the dual (D). 


\subsection{Piecewise Concave Test Functions}\label{piecewise concave}

When $\phi_s|_{\Xi_i}$ and $-h|_{\Xi_i}$ are jointly concave but not linear, then 
$\varphi_{y,z}(x)$ is concave and $\Xi_i=\conv(q^i_1,\dots,q^i_{n_i})$ is a polytope. 
The copositivity constraint 
\begin{equation}\label{copos}
\sum_{s=1}^{M}y_s\phi_s(x)+\sum_{t=1}^{N}z_t\psi_t(x)+z_0\1_{\Phi}(x)
-h(x)\geq 0,\quad(x\in\Xi_i) 
\end{equation}
can then be written as 
\begin{equation*}
\varphi_{y,z}(q^i_j)\geq 0,\quad(j=1,\dots,n_i). 
\end{equation*}
Thus, \eqref{copos} can be replaced by $n_i$ linear inequality constraints on the decision variables $y_s$ and $z_t$. 


\subsection{Piecewise Polynomial Test Functions}\label{piecewise polynomial}

Another case that can be treated via finitely many constraints is when $\phi_s|_{\Xi_i}$,  $\psi_t|_{\Xi_i}$ and $h|_{\Xi_i}$ are jointly polynomial. The approach of Lasserre \cite{lasserre} and Parrilo \cite{parrilo} can be applied to turn the copositivity constraint 
\begin{equation*}
\langle f^i_j, x\rangle \geq \ell^i_j,\quad (j=1,\dots,k_i)\Longrightarrow
\varphi_{y,z}(x)\geq 0, 
\end{equation*}
into finitely many linear matrix inequalities. However, this approach is generally limited to low dimensional applications. \\


\section{Conclusions}
Our analysis shows that a wide range of duality relations in use in quantitative risk management can be understood from the single perspective of a generalized duality relation discussed in Section \ref{duality}. An interesting class of special cases is provided by formulating a finite number of constraints in the form of bounds on integrals. The duals of such models are semi-inifinite optimization problems that can often be reformulated as finite optimization problems, by making use of standard results on co-positivity. \\

\noindent  {\bf Acknowledgements:} 
Part of this research was conducted while the first author was visiting the 
FIM at ETH Zurich during sabbatical leave from Oxford. He thanks for the 
support and the wonderful research environment he encountered there. 
The research of the the first two authors was supported through grant EP/H02686X/1 
from the Engineering and Physical Sciences Research Council of the UK.

\end{document}